\documentclass[11pt, a4paper]{article}
\usepackage{amssymb,amsmath,amsthm,bbm}
\usepackage{enumitem}
\usepackage[bookmarksopen=true,pagebackref=true]{hyperref}
\usepackage{cleveref}
\usepackage{tikz}
\usepackage{setspace}
\usepackage[margin=1in]{geometry}

\newtheorem{theorem}{Theorem}[section]

\newtheorem{lemma}[theorem]{Lemma}

\theoremstyle{definition}
\newtheorem{defn}[theorem]{Definition}

\def \szip {\textsc{Szip}}

\newcommand{\cL}{\mathcal{L}}
\newcommand{\bE}{\mathbb{E}}

\makeatletter
\def\@fnsymbol#1{\ensuremath{\ifcase#1\or \dagger\or \ddagger\or
		\mathsection\or \mathparagraph\or \|\or **\or \dagger\dagger
		\or \ddagger\ddagger \else\@ctrerr\fi}}
\makeatother

\begin{document}
\onehalfspace

\title{Structural Entropy of the Stochastic Block Models} 
\author{Jie Han \thanks{J. Han, T. Guo, W. Han, B. Bai, and G. Zhang are with Theory Lab, Central Research Institute, 2012 Labs, Huawei Tech. Co., Ltd.; (han.jie@huawei.com, guo.tao1@huawei.com, harvey.hanwei@huawei.com, baibo8@huawei.com, nicholas.zhang@huawei.com)} , 
	Tao Guo $^{\dagger}$, 
	Qiaoqiao Zhou \thanks{Q. Zhou is with the Department of Computer Science, School of Computing, National University of Singapore; (zhouqq@comp.nus.edu.sg)} , 
	Wei Han $^{\dagger}$, Bo Bai $^{\dagger}$, and Gong Zhang $^{\dagger}$
}

\date{}
\maketitle

\begin{abstract}
	With the rapid expansion of graphs and networks and the growing magnitude of data from all areas of science, effective treatment and compression schemes of context-dependent data is extremely desirable.
	A particularly interesting direction is to compress the data while keeping the ``structural information'' only and ignoring the concrete labelings.
	Under this direction, Choi and Szpankowski introduced the structures (unlabeled graphs) which allowed them to compute the structural entropy of the Erd\H{o}s--R\'enyi random graph model.
	Moreover, they also provided an asymptotically optimal compression algorithm that (asymptotically) achieves this entropy limit and runs in expectation in linear time.
	
	In this paper, we consider the Stochastic Block Models with an arbitrary number of parts.
	Indeed, we define a partitioned structural entropy for Stochastic Block Models, which generalizes the structural entropy for unlabeled graphs and encodes the partition information as well.
	We then compute the partitioned structural entropy of the Stochastic Block Models, and provide a compression scheme that asymptotically achieves this entropy limit.
\end{abstract}


\section{Introduction}
Shannon's metric of ``Entropy" of information is a foundational concept of information theory~\cite{Raymond-book,Thomas-Cover-book}.
Given a discrete random variable $X$ with support set (that is, the possible outcomes) $x_1, x_2, \dots, x_n$, which occurs with probability $p_1, p_2, \dots, p_n$, the entropy of $X$ is defined as
\[
H(X):= - \sum_{i=1}^m p_i \log p_i,
\]
where the logarithm here and throughout this paper is of base 2.
Note that the entropy of $X$ is a function of the probability distribution of $X$.

The entropy was originally created by Shannon in~\cite{Shannon-48} as part of his theory of communication, where a data communication system consists of a data source $X$, a channel and a receiver.
The fundamental problem of communication is for the receiver to reliably recover what data was generated by the source, based on the bits it receives through the channel.
Shannon proved that the entropy of the source $X$ plays a central role -- in his source coding theorem it is shown that the entropy is the mathematical limit on how well the data can be losslessly compressed.

The question then arises: \emph{How to compress data that has structures, e.g., data in social networks?}
In Shannon's 1953 less known paper~\cite{Shannon-53} he argued for an extension of information theory, where data is considered as observations of a source, to ``non-conventional data'' (that is, lattices).
Indeed, nowadays data appears in various formats and structures (e.g., sequences, expressions, interactions) and in drastically increasing amounts.
In many scenarios, data is highly context-dependent and in particular, the structural information and the context information seem to be two conceptually different aspects.
Therefore it is desirable to develop novel theory and efficient algorithms for extracting useful information from non-conventional data structures.
Roughly speaking, such data consists of structural information, which, might be understood as the ``shape'' of the data, and context information which should be recognized as data labels.

It is well-known that complex networks (e.g., social networks) admit community structures~\cite{palla2005uncovering}. That is, users within a group interact with each other more frequently than those outside the group. The Stochastic Block Model (SBM) \cite{SBM1983} is a celebrated random graph model that has been widely used to study the community structures in graphs and networks. It provides a good benchmark to evaluate the performance of community detection algorithms and inspires the design of many algorithms for community detection tasks. The theoretical underpinnings of the SBM have been extensively studied and sharp thresholds for exact recovery have been successively established~\cite{abbe2015award,mossel2015,abbe2015multiple-community,hajek2017hidden}. We refer readers to~\cite{abbe2017survey} for a recent survey, where other interesting and important problems in SBM are also discussed.  

In addition to the SBM model discussed in \cite{abbe2017survey}, there are other angles to study compression of data with graph structures. 
Asadi et. al. \cite{survey-02} investigated data compression on graphs with clusters. 
Zenil et. al. \cite{survey-01} have surveyed information-theoretic methods, in particular Shannon entropy and algorithmic complexity, for characterizing graphs and networks. 

\subsection{Compression of graphs}
In recent years, graphical data and the network structures supporting them are becoming increasingly common and important in branches of engineering and sciences.
To better represent and transmit graphical data, many works consider the problem of compressing the (random) graph up to isomorphism, i.e., compressing the structure of a graph.  
A graph $G$ contains a finite set $V$ of vertices and a set $E$ of edges each of which connects two vertices.
A graph can be represented by a binary matrix (the adjacency matrix) that further can be viewed as a binary sequence.
Thus, encoding a labeled graph (that is, all vertices need to be distinguished) is equivalent to encoding the $\binom{|V|}2$-digit binary sequence, given certain probability distribution on all $\binom{|V|}2$ possible edges.
However, such a string does not reflect internal symmetries that are conveyed by the graph automorphism, and sometimes we are only interested in the local or global structures in the graph, rather than the exact vertex labelings.
The structural entropy is defined when the graphs are considered unlabeled, or simply called structures, where the vertices are viewed as undistinguishable.
The goal of this natural definition is to capture the information of the structure, and thus provides a fundamental measure in graph/structure compression schemes.

The problem actually has a strong theoretical background.
Back to 1984, Tur\'an~\cite{Turan-84} raised the question of finding an efficient coding method for general unlabeled graphs on $n$ vertices, where a lower bound of $\binom n2 - n\log n + O(n)$ bits is suggested.
This lower bound can be seen by the number of unlabeled graphs~\cite{Harary-73}.
The question was later answered by Naor~\cite{Naor-90} in 1990 who proposed such a representation that is optimal up to the first two leading terms when all unlabeled graphs are equally likely.
In a recent paper Kieffer et al.~\cite{Kieffer-09} proved a structural complexity of a binary tree.
There also have been some heuristic methods for real-world graph compression schemes, see~\cite{Adler-01, Chierichetti-09, Peshkin-07,Savari-04, Sun-07}.
Rather recently, Choi and Szpankowski~\cite{Choi-Wojciech-12} studied the structural entropy of the Erd\H{o}s--R\'enyi random graph $\mathcal G(n, p)$.
They computed the structural entropy given that $p$ is not (very) close to 0 or 1 and also gave a compression scheme that matches their computation.
Later, the structural entropy for other randomly generated graphs, e.g. the preferential attachment graphs and web graphs are also studied~\cite{Luczak-19-conf,Luczak-19-journal,Sauerhoff-16,Kontoyiannis-21}.

However, it is well-known that the  Erd\H{o}s--R\'enyi model is too simplistic to model real networks, in particular due to its strong homogeneity and absence of community structure.  
In this paper, we consider the compression of graphical structures of the SBM, which in general model real networks better and circumvent the issues of the ER-model.  
In summary, our contributions are as follows:
\begin{itemize}
	\item We introduce the partitioned structural entropy which generalizes the structural entropy for unlabeled graphs and we show that it reflects the partition information of the SBM.
	\item We provide an explicit formula for the partitioned structural entropy of the SBM. 
	\item We also propose a compression scheme that asymptotically achieves this entropy limit.
\end{itemize}

Semantic communications are considered as a key component of future generation networks, where a natural problem to consider is how to efficiently extract and transmit the ``semantic information''.
In the case of graph data, one may view the (partitioned) structures as the information that needs to be abstracted while the concrete labeling information is considered redundant.
From this point of view, our result is a step for the study of semantic compression/communication under appropriate contexts.

\subsection{Related works}
Finally, we would like to point out that there are some other information metrics defined on graphs. 
The term ``graph entropy'' has been defined and used in the history.
For example, graph entropy introduced by K\H{o}rner in \cite{korner1973} denotes the number of bits one has to convey to resolve the ambiguity of a vertex in a graph. This notion also turns out to be useful in other areas, including combinatorics.  Chromatic entropy introduced in \cite{alon1996} is the lowest entropy of any coloring of a graph. It finds application in zero-error source coding. 
We remark that the structural entropy we considered is quite different from the K\H{o}rner graph entropy and chromatic entropy. 

On the other hand, a concept of graph entropy (also called \emph{topological information content of a graph}) was introduced by Rashevsky~\cite{Rashevsky-55} and Trucco~\cite{Trucco-56}, and later by Mowshowitz~\cite{Mowshowitz-68-1,Mowshowitz-68-2,Mowshowitz-68-3,Mowshowitz-68-4,Mowshowitz-12,Dehmer-11}, which is defined as a function of (the structure of) a graph and an equivalence relation defined on its vertices or edges.
Such a concept is a measure of the graph itself and does not involve any probability distribution.

\section{Preliminaries}
\subsection{Structural entropy of unlabeled graphs}
Now let us formally define the structural entropy given a probability distribution on unlabeled graphs.

Given an integer $n$, define $\mathcal G_n$ as the collection of all $n$-vertex labeled graphs.
\begin{defn}[Entropy of Random Graph]
	Given an integer $n$ and a probability distribution on $\mathcal G_n$, the entropy of a random graph $\mathcal G\in\mathcal G_n$ is defined as
	\[
	H_{\mathcal G} = \mathbb E[-\log P(G)] = - \sum_{G\in \mathcal G_n} P(G) \log P(G)
	\]
	where $P(G)\triangleq P(\mathcal{G}=G)$ is the probability of a graph $G$ in $\mathcal G_n$.
\end{defn}
Then the random structure model $\mathcal S_n$ associated with the probability distribution $\mathcal G_n$, is defined as the unlabeled version of $\mathcal G_n$.
For a given $S\in \mathcal S_n$, the probability of $S$ can be computed as
\[
P(S)=\sum_{G\cong S, G\in \mathcal G_n}P(G).
\]
Here $G\cong S$ means that $G$ and $S$ have the same structure, that is, $S$ is isomorphic to $G$.
Clearly if all isomorphic labeled graphs have the same probability, then for any labeled graph $G\cong S$, one has
\[
P(S)=N(S)\cdot P(G)
\]
where $N(S)$ stands for the number of different labeled graphs that have the same structure as~$S$.
\begin{defn}[Structural Entropy]
	The structural entropy $H_{\mathcal S}$ of a random graph $\mathcal G$ is defined as the entropy of a random structure $\mathcal S$ associated with $\mathcal G_n$, that is,
	\[
	H_{\mathcal S}=\mathbb E[-\log P(S)] = - \sum_{S\in \mathcal S} P(S) \log P(S)
	\]
	where the sum is over all distinct structures.
\end{defn}

The Erd\H{o}s--R\'enyi random graph $\mathcal G(n, p)$, also called the binomial random graph, is a fundamental random graph model, which has $n$ vertices and each pair of vertices is connected with probability $p$, independent of other pairs.
In 2012, Choi and Szpankowski~\cite{Choi-Wojciech-12} proved the following for the Erd\H{o}s--R\'enyi random graphs.

\begin{theorem}
	[Choi and Szpankowski,~\cite{Choi-Wojciech-12}]
	\label{thm:CSz12}
	For large $n$ and all $p$ satisfying $n^{-1}\ln n\ll p$ and $1-p\gg n^{-1}\ln n$, the following holds:
	\begin{enumerate}
		\item The structural entropy $H_{\mathcal S}$ of $\mathcal G(n,p)$ is
		\[
		H_{\mathcal S} = \binom n2 h(p) - \log n! + O\left( \frac{\log n}{n^{\alpha}} \right)
		\]
		for some $\alpha>0$.
		\item For a structure $S$ of $n$ vertices and $\varepsilon >0$
		\[
		P\left( \left| -\frac{1}{\binom n2}\log P(S) - h(p) + \frac{\log n!}{\binom n2} \right| < \varepsilon \right) > 1-2\varepsilon
		\]
		where $h(p)=-p\log p - (1-p)\log (1-p)$ is the entropy rate of a binary memoryless source.
	\end{enumerate}
\end{theorem}

Furthermore, they~\cite{Choi-Wojciech-12}  also presented a compression algorithm for unlabeled graphs that asymptotically
achieves the structural entropy up to an $O(n)$ error term.
%

\subsection{Stochastic Block Model -- Our result}
It is well-known that the ER model is too simplistic to model real networks, in particular due to its strong homogeneity and absence of community structure.
The Stochastic Block Model is then introduced on the assumption that vertices in a network connect independently but with probability based on their profiles, or equivalently, on their community assignment.
For example, in the SBM with two communities and symmetric parameters, also known as the planted bisection model, denoted by $\mathcal G(n,p,q)$, the vertex set is partitioned into two sets $V_1$ and $V_2$, any pair of vertices inside $V_1$ or $V_2$ are connected with probability $p$ and any pair of vertices across the clusters are connected with probability $q$, and all these connections are independent.

As an illuminating example, consider a context $G$ where there are $n/2$ users and $n/2$ devices, and each pair of users and each pair of devices are connected with probability $p$, a user and a device is connected with probability $q$ and each of these connections is independent of all other connections.
Suppose that we need to compress the information of $G$.
However, in the context it is not appropriate to view $G$ as an unlabeled graph, that is, in addition to the structure information, it is also important to keep the ``community'' information -- the compression also needs to encode the information that who is a user and who is device.


\begin{defn}[Partition-respecting isomorphism, Partitioned Unlabeled Graphs]
	Let $r\le n$ be integers.
	Suppose $V$ is a set of $n$ vertices and $\mathcal P=\{V_1, V_2, \dots, V_r\}$ is a partition of $V$ into $r$ parts.
	The partition-respecting isomorphism, denoted by ``$\cong_{\mathcal P}$'' is defined as follows.
	For any two labeled graphs $G$ and $G'$, we write $G\cong_{\mathcal P} G'$ if and only if $G\cong G'$ they are isomorphic via an isomorphism function $\phi: V\to V$ such that $\phi(V_i)=V_i$, for $1\le i\le r$.
	%
	Then $\Gamma_{\mathcal P}$ is defined as the collection of $n$-vertex graphs on $V$ where we ignore the labels of vertices inside each $V_i$, $1\le i\le r$, namely, the equivalence classes under partition-respecting isomorphism, with respect to $\mathcal P$.
	
\end{defn}

Note that every labeled graph $G$ corresponds to a unique structure $S\in \Gamma_\mathcal P$, and we use $G\cong_{\mathcal P} S$ to denote this relation.
Furthermore, under the above definition, general unlabeled graphs correspond to the case $r=1$.

\begin{defn}[Partitioned Structural Entropy]
	Let $V$ be a set of $n$ vertices where $n\in \mathbb N$.
	Suppose $\mathcal P=\{V_1, V_2, \dots, V_r\}$ is a partition of $V$ into $r$ parts and $\mathcal S_n$ is a probability distribution over all partitioned unlabeled graphs on $n$ vertices.
	Then the structural entropy $H_{\mathcal S}$ associated to $\mathcal S_n$ is defined by
	\[
	H_{\mathcal S}=\mathbb E[-\log P(S)] = - \sum_{S\in {\mathcal S_n}} P(S) \log P(S).
	\]

\end{defn}

In this paper we extend Theorem~\ref{thm:CSz12} to the structural entropy of the Stochastic Block Model with any given number of blocks, and provide a compression algorithm that asymptotically matches this structural entropy.
We first give the result for the balanced bipartition case $\mathcal G(n,p,q)$.

\begin{theorem}\label{thm:SBM}
Let $V=V_1\cup V_2$ be a set of $n$ vertices and $|V_1|=|V_2|=n/2$.
Suppose $\mathcal G(n,p, q)$ is a probability distribution of graphs on $V$ where every edge inside $V_1$ or $V_2$ is present with probability $p$ and every edge between $V_1$ and $V_2$ is present with probability $q$, and these edges are mutually independent.
	For large even $n$ and all $p$ satisfying $n^{-1}\ln n\ll p, q$ and $1-p\gg n^{-1}\ln n$, the following holds:
	\begin{enumerate}
		\item[(i)] The partitioned structural entropy $H_{\mathcal S}$ of $\mathcal G(n,p, q)$ is
		\begin{equation}
		H_{\mathcal S} = 2\binom {n/2}2 h(p) + \frac{n^2}{4}h(q) - 2\log \left(\frac n2\right)! + O\left( \frac{\log n}{n^{\alpha}} \right)  \label{stru-entropy-value}
		\end{equation}
		for some $\alpha>0$.
		\item[(ii)] For a balanced bipartitioned structure $S$ and $\varepsilon >0$
		\[
		P\left( \left| -\frac{1}{\binom n2}\log P(S) - \frac{n-2}{2n-2} h(p) - \frac{n}{2n-2}h(q) + \frac{2\log (n/2)!}{\binom n2} \right| < 3\varepsilon \right) > 1 - 4\varepsilon
		\]
		where $h(p)=-p\log p - (1-p)\log (1-p)$ is the entropy rate of a binary memoryless source.
	\end{enumerate}
\end{theorem}

Note that the structural entropy $H_{\mathcal S}$ here is larger than that in Theorem~\ref{thm:CSz12} (even if $p=q$), which reflects the fact that the SBM with ``a planted (bi-)partition'' contains \emph{prefixed} structures, so has less symmetries than $\mathcal G(n,p)$, the pure random model\footnote{For $\mathcal G(n,p)$, when it is asymmetric, comparing with the completely labeled graphs, Theorem~\ref{thm:CSz12} saves a term as $\log n!$; this saving becomes $2\log \left( n/2\right)!$ for the planted balanced bipartition case in Theorem~\ref{thm:SBM}.}.

\section{Proof of Theorem~\ref{thm:SBM}}

One key ingredient in the proof of Theorem~\ref{thm:CSz12} in~\cite{Choi-Wojciech-12} is the following lemma on the symmetry of $\mathcal G(n,p)$.
A graph is called asymmetric if its automorphism group does not contain any permutation other than identity;
otherwise it is called symmetric.

\begin{lemma}
	[Kim, Sudakov and Vu, 2002]\label{lem:KSV02}
	For all $p$ satisfying $n^{-1}\ln n\ll p$ and $1-p\gg n^{-1}\ln n$, a random graph $G\in \mathcal G(n,p)$ is symmetric with probability $O(n^{-w})$ for any positive constant $w$.
\end{lemma}

\begin{proof}[Proof of Theorem~\ref{thm:SBM}]
	Note that every pair of vertices in $V_1$ or in $V_2$ should be considered as undistinguishable, but not the pairs of vertices in $X\times Y$.
	Recall that we write $G\cong_{\mathcal P} S$ for a graph $G$ and a structure $S$ if $S$ represents the structure of $G$ (with respect to the partition $\mathcal P$).
	
	Let $\mathcal G:=\mathcal G(n,p,q)$.
	We first compute $H_{\mathcal G}$.
	Note that there are $\binom n2$ possible edges in $G\in \mathcal G$, and we can view it as a binary sequence of length $\binom n2$, where each digit is a Bernoulli random variable.
	Moreover, for edges inside $V_1$ or $V_2$, the random variable, denoted by $X_1$, has expectation $p$ and for edges in $V_1\times V_2$ the random variable, denoted by $X_2$, has expectation $q$. Thus we have
	\begin{align*}
	H_{\mathcal G} & = - \mathbb E [\log X_1^{2\binom {n/2}2} X_2^{n^2/4}] \\
	& = - 2\binom {n/2}2  \mathbb E [\log X_1] - \frac{n^2}{4}  \mathbb E [\log X_2] \\
	& = 2\binom {n/2}2 h(p) + \frac{n^2}{4} h(q).
	\end{align*}
	
	Now write $\mathcal S_n$ for the probability distribution on $V$ over all partitioned unlabeled graphs inherited from $\mathcal G$, namely, given $S\in \Gamma_\mathcal P$, $P(S)=\sum_{G\cong_{\mathcal P} S} P(G)$.
	Let $H_{\mathcal S}$ be the partitioned structural entropy of $\mathcal S_n$.
	Therefore, compared with our goal, it remains to show that
	\begin{equation}\label{eq:SG}
	H_{\mathcal S} - H_{\mathcal G} = - 2\log \left(n/2\right)! + O\left( \frac{\log n}{n^{\alpha}} \right).
	\end{equation}
	
	Note that in $\mathcal G(n,p, q)$, all labeled graphs $G\in \mathcal G$ such that $G\cong_{\mathcal P} S$ have the same probability $P(G)$.
	Thus, given a (labeled) graph $G\in \mathcal G$, we have $P(G)= P(S)/ N(S)$, where $S\in \mathcal S_n$ is such that $G\cong_{\mathcal P} S$.
	So the graph entropy of $\mathcal G=\mathcal G(n,p,q)$ can be written as
	\begin{align}
	H_{\mathcal G} &=  - \sum_{G\in \mathcal G} P(G) \log P(G) \nonumber \\
	&=  - \sum_{S\in \mathcal S_n} \sum_{G\cong_{\mathcal P} S, G\in \mathcal G} P(G) \log P(G)\nonumber \\
	&=  - \sum_{S\in \mathcal S_n} \sum_{G\cong_{\mathcal P} S, G\in \mathcal G} \frac{P(S)}{N(S)} \log  \frac{P(S)}{N(S)} \nonumber \\
	&=  - \sum_{S\in \mathcal S_n} P(S) \log  \frac{P(S)}{N(S)} \nonumber \\
	&=  H_{\mathcal S} + \sum_{S\in \mathcal S} P(S) \log N(S) \label{eq:GS}
	\end{align}
	Define $S[W]$ be be $S$ restricted on $W$ for $W\in V$. 
	Now we split $S$ into $S_1$ and $S_2$, i.e., $S_1=S[V_1]$ and $S_2=S[V_2]$.
	Write $Aut(S_i)$ for the automorphism group for $S_i$, and we naturally have
	\[
	N(S) = \frac{(n/2)!\cdot (n/2)!}{|Aut(S_1)||Aut(S_2)|}.
	\]
	Combining this with~\eqref{eq:SG} and~\eqref{eq:GS}, it remains to show that
	\[
	\sum_{S\in \mathcal S} P(S) \log|Aut(S_1)||Aut(S_2)| = O\left( \frac{\log n}{n^{\alpha}} \right).
	\]
	In the summation above we only need to focus on $S$ such that either $S_1$ or $S_2$ is symmetric, as otherwise $\log|Aut(S_1)||Aut(S_2)|=\log 1=0$.
	By Lemma~\ref{lem:KSV02}, we conclude that the probability of $S$ restricted on $V_1$ or $V_2$ is symmetric is $O(n^{-1-\alpha})$ for some $\alpha>0$, and for such $S$ we use the trivial bound $\log|Aut(S_1)||Aut(S_2)|\le 2\log (n/2)!\le 2n\log n$.
	This gives us the desired estimate in (i)
	\[
	\sum_{S\in \mathcal S} P(S) \log|Aut(S_1)||Aut(S_2)| \le 2n\log n \cdot O(n^{-1-\alpha}) = O\left( \frac{\log n}{n^{\alpha}} \right).
	\]
	
	To show (ii), for a set $V$ of $n$ vertices and a balanced bipartition $\mathcal P=(V_1, V_2)$ of $V$, we define the typical set $T_{\varepsilon}^n$ as the set of structures $S$ on $n$ vertices satisfying (a) $S$ is asymmetric on $V_1$ and $V_2$, respectively; (b) for $G\cong_{\mathcal P} S$
	\[
	2^{-2\binom{n/2}2 h(p) - \tfrac{n^2}{4} h(q) - \binom n2 \varepsilon} \le P(G) \le 2^{-2\binom{n/2}2 h(p) - \tfrac{n^2}{4} h(q) + \binom n2 \varepsilon}.
	\]
	Denote by $T_{1}^n$ and $T_{2}^n$ the sets of structures satisfying the properties (a) and (b), respectively and thus we have $T_{\varepsilon}^n=T_1^n\cap T_2^n$.
	Firstly, by the asymmetry of $\mathcal G(n,p)$ (Lemma~\ref{lem:KSV02}), we conclude that $P(T_1^n) > 1-2\varepsilon$ for large $n$.
	Secondly, we use a binary sequence of length $\binom n2$ to represent a (labeled) instance $G$ of $\mathcal G(n,p, q)$, where the first $\binom{n/2}2$ bits $\cL_1$ represent the induced subgraph on $V_1$, the next $\binom{n/2}2$ bits $\cL_2$ represent the induced subgraph on $V_2$, and finally the rest $n^2/4$ bits $\cL_{12}$ represent the bipartite graph on $V_1\times V_2$.
	Since all edges of $G$ are generated independently, both $\cL_1$ and $\cL_2$ have in expectation $\binom{n/2}2 p$ 1's and the AEP property of the binary sequences implies that
	\[
	2^{-\binom{n/2}2 h(p) - \binom n2 \varepsilon} \le P(G[V_1]), P(G[V_2]) \le 2^{-\binom{n/2}2 h(p) + \binom n2 \varepsilon}
	\]
	holds with probability at least $1-2\varepsilon$.
	Similarly, $\cL_{12}$ has in expectation $(n^2/4)q$ 1's and the AEP property of the binary sequences gives that with probability at least $1-\varepsilon$,
	\[
	2^{- \tfrac{n^2}{4} h(q) - \binom n2 \varepsilon} \le P(G[V_1, V_2]) \le 2^{- \tfrac{n^2}{4} h(q) + \binom n2 \varepsilon}
	\]
	Since these edges are independent, we finally conclude that (b) holds with probability at least $1-3\varepsilon$.
	Thus, $P(T_{\varepsilon}^n)\ge 1-4\varepsilon$.
	Now we can compute $P(S)$ for $S\in T_{\varepsilon}^n$.
	By (a), $P(S)=(n/2)!(n/2)!P(G)$ for any $G\cong S$.
	Together with (b) and straightforward computation, the assertion of (ii) follows.
\end{proof}

\section{SBM Compression Algorithm}
Given the computation of the structural entropy, a natural next step is to design efficient compression schemes that are close to or even (asymptotically) achieve this entropy limit.
Choi and Szpankowski~\cite{Choi-Wojciech-12} presented such an algorithm (which they named \szip) for (unlabeled) random graphs, which uses in expectation at most $\binom n2 h(p) - n\log n + O(n)$ bits and asymptotically achieves the structural entropy given in Theorem~\ref{thm:CSz12}.
Roughly speaking, \szip~greedily peels off vertices from the graph and (efficiently) store the neighborhood information.
This procedure can be simply reversed but the labeling of the recovered graph may be different from the original graph, which is the reason on why a saving of the codeword length is achieved. 
Refinements and analysis~\cite{Choi-Wojciech-12} are also provided to achieve the proposed performance.

Here we give an algorithm that optimally compresses SBMs which uses the \szip~algorithm as building blocks and matches the structural entropy computation in Theorem~\ref{thm:SBM}.
The algorithm consists of two stages. 
It first compresses $S[V_1]$ and $S[V_2]$ using \szip~and then compresses $S[V_1,V_2]$ using an arithmetic compression algorithm with the help of \szip~decoding outputs.
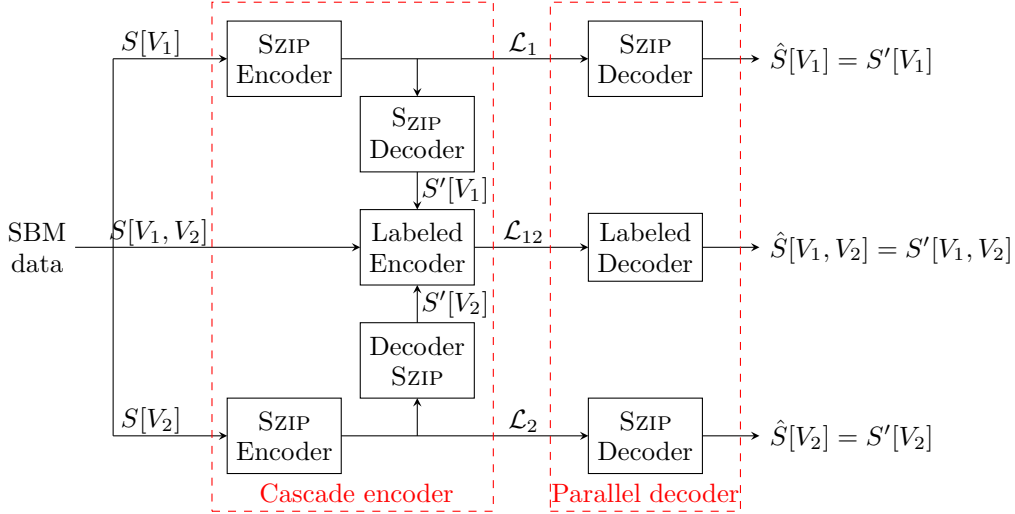
\begin{figure}[t]
	\centering
	\begin{tikzpicture}[scale=1.0, font=\small]
	\node at (0,0.2) {SBM};
	\node at (0,-0.2) {data};
	\draw (0.5,0)--(1.0,0);
	\draw (1.0,2.5)--(1.0,-2.5);
	
	\node at (1.5,2.7) {$S[V_1]$};
	\draw[->,>=stealth] (1.0,2.5)--(2.5,2.5);
	\node at (3.25,2.7) {\szip};
	\node at (3.25,2.3) {Encoder };
	\draw (2.5,2.0) rectangle (4.0,3.0);
	\node at (6.4,2.7) {$\cL_1$};
	\draw[->,>=stealth] (4.0,2.5)--(7.25,2.5);
	
	\node at (1.5,-2.3) {$S[V_2]$};
	\draw[->,>=stealth] (1.0,-2.5)--(2.5,-2.5);
	\node at (3.25,-2.3) {\szip};
	\node at (3.25,-2.7) {Encoder };
	\draw (2.5,-2.0) rectangle (4.0,-3.0);
	\node at (6.4,-2.3) {$\cL_2$};
	\draw[->,>=stealth] (4.0,-2.5)--(7.25,-2.5);
	
	\node at (1.6,0.2) {$S[V_1,V_2]$};
	\draw[->,>=stealth] (1.0,0)--(4.25,0);
	
	\draw[->,>=stealth] (5.0,2.5)--(5.0,2.0);
	\node at (5.0,1.7) {$\text{S}_{\text{ZIP}}$};
	\node at (5.0,1.3) {Decoder };
	\draw (4.25,1.0) rectangle (5.75,2.0);
	
	\draw[->,>=stealth] (5.0,-2.5)--(5.0,-2.0);
	\node at (5.0,-1.7) {\szip};
	\node at (5.0,-1.3) {Decoder };
	\draw (4.25,-1.0) rectangle (5.75,-2.0);
	
	\node at (5.5,0.75) {$S'[V_1]$};
	\draw[->,>=stealth] (5.0,1.0)--(5.0,0.5);
	\node at (5.5,-0.75) {$S'[V_2]$};
	\draw[->,>=stealth] (5.0,-1.0)--(5.0,-0.5);
	\node at (5.0,0.2) {Labeled };
	\node at (5.0,-0.2) {Encoder };
	\draw (4.25,-0.5) rectangle (5.75,0.5);
	\node at (6.4,0.2) {$\cL_{12}$};
	\draw[->,>=stealth] (5.75,0)--(7.25,0);
	
	\draw [dashed,red](2.3,-3.5) rectangle (6.0,3.25);
	\node [red] at (4.2,-3.25) {Cascade encoder};
	
	\node at (8.0,0.2) {Labeled };
	\node at (8.0,-0.2) {Decoder };
	\draw (7.25,-0.45) rectangle (8.75,0.45);
	\draw[->,>=stealth] (8.75,0)--(9.5,0);
	\node [right] at (9.5,0) {$\hat{S}[V_1,V_2]=S'[V_1,V_2]$};
	
	\node at (8.0,2.7) {\szip};
	\node at (8.0,2.3) {Decoder };
	\draw (7.25,2.0) rectangle (8.75,3.0);
	\draw[->,>=stealth] (8.75,2.5)--(9.5,2.5);
	\node [right] at (9.5,2.5) {$\hat{S}[V_1]=S'[V_1]$};
	
	\node at (8.0,-2.3) {\szip};
	\node at (8.0,-2.7) {Decoder };
	\draw (7.25,-2.0) rectangle (8.75,-3.0);
	\draw[->,>=stealth] (8.75,-2.5)--(9.5,-2.5);
	\node [right] at (9.5,-2.5) {$\hat{S}[V_2]=S'[V_2]$};
	
	\draw [dashed,red](6.75,-3.5) rectangle (9.25,3.25);
	\node [red] at (8.0,-3.25) {Parallel decoder};
	
	\end{tikzpicture}
	\caption{Illustration of compression algorithm}
	\label{fig_compression}
\end{figure}

To give a brief description of the compression algorithm, we again use the balanced bipartition $V_1\cup V_2$ as an example.
The encoding and decoding procedure of the algorithm is illustrated in Figure~\ref{fig_compression}. 
The algorithm encodes the observed $\mathcal S(n,p,q)$ into a binary string as follows.
It uses \szip~as a subroutine to compress $S[V_1]$ and $S[V_2]$ into binary sequences $\cL_1$ and $\cL_2$. 
Then, as part of the encoder, we run the \szip~decoder on $\cL_1$ and $\cL_2$ to obtain decoded structures $S'[V_1]$ and $S'[V_2]$, respectively. 
We then compress $S[V_1, V_2]$ as a labeled bipartite graph under the vertex labeling of $S'[V_1]$ and $S'[V_2]$ into $\cL_{12}$. 
This ``Labeled Encoder" can be done by treating it as a binary sequence of length $n^2/4$ and compressing using a standard arithmetic encoder \cite{Arithmetic-01,Arithmetic-02,Arithmetic-1}. 
The concatenation of \szip~algorithms and the arithmetic encoder forms the cascade encoder of our algorithm and obtains the codeword $(\cL_1,\cL_2,\cL_{12})$. 
Upon receiving the codeword, we decode them parallelly using \szip~decoder and the arithmetic decoder. 
This completes our algorithm. 

The main challenge in the design of our algorithm is how the decoder can retrieve the consistency between the bipartite graph $S[V_1, V_2]$ and the decoded version of $S[V_1]$ and $S[V_2]$.
A key observation here is that since \szip~is a deterministic algorithm, 
although it may permute the vertex labelings, its output is an invariant given the same input. 
Given this, our solution here is to first run \szip~(both encoding and decoding) at the encoder, 
and obtain structures $S'[V_1]$ and $S'[V_2]$, respectively.
We then compress $S[V_1, V_2]$ (as a labeled bipartite graph) under the vertex labeling of $S'[V_1]$ and $S'[V_2]$.
This would guarantee that the decoded structures $\hat{S}[V_1]$, $\hat{S}[V_2]$ and $\hat{S}[V_1, V_2]$ share the same vertex labeling as $S'[V_1]$ and $S'[V_2]$, namely, $S$ is recovered.

Before discussing the performance of the algorithm, we first describe some useful properties of the arithmetic compression algorithm in the following lemma. 
We omit the proof of the lemma, which follows from the analysis in \cite{Arithmetic-01,Arithmetic-02,Arithmetic-1} and AEP properties in \cite{Raymond-book,Thomas-Cover-book}. 
\begin{lemma}\label{lemma-arithmetic}
	Let $L$ be the codeword length of the arithmetic compression algorithm when compressing a binary sequence with length $m$ and entropy rate $h$. 
	For large $m$, the following holds:
	\begin{enumerate}[label=$(\roman*)$]
		\item The expected codeword length asymptotically achieves the entropy of the message, i.e., 
		\begin{equation}
		\bE[L]= mh + O(\log m). \label{arithmetic-length}
		\end{equation}
		
		\item For any $\epsilon>0$,
		\begin{equation}
		P(|L-\bE[L]|\leq \epsilon \log m)\geq 1-o(1).
		\end{equation}
		
		\item The arithmetic algorithm runs in time $O(m)$. 
	\end{enumerate}
\end{lemma}

The following theorem characterizes the performance of our algorithm. 
It is immediate from Theorem 2 in \cite{Choi-Wojciech-12} (performance of \szip) and \Cref{lemma-arithmetic}, we omit the detailed proofs here. 
\begin{theorem}\label{thm-alg-performance}
	Let $V=V_1\cup V_2$ be a set of $n$ vertices and $|V_1|=|V_2|=n/2$.
	Given a partitioned unlabeled graph $S$ on $V$, let $L(S)$ be the codeword length given by our algorithm.
	For large $n$, our algorithm runs in time $O(n^2)$, and satisfies the following:
	\begin{enumerate}[label=$(\roman*)$]
		\item The algorithm asymptotically achieves the structural entropy in~\eqref{stru-entropy-value} \footnote{Note that $(n/2)\log (n/2) = n\log n + O(n)$.}, i.e., 
		\begin{equation*}
		\bE[L(S)]\le 2\binom{n/2}2 h(p) + \frac{n^2}4 h(q) - n\log n +O(n). \label{alg-length}
		\end{equation*}
		
		\item For any $\epsilon>0$,
		\begin{equation*}
		P(|L(S)-\bE[L(S)]|\leq \epsilon n\log n)\geq 1-o(1).
		\end{equation*}
		
	\end{enumerate}
\end{theorem}

\section{General SBM with $r\ge 2$ blocks}
In previous sections, we discussed the structural entropy of SBM and the compression algorithm that asymptotically achieves this structural entropy for the balanced bipartition case ($r=2$). 
The corresponding results in \Cref{thm:SBM} and \Cref{thm-alg-performance} can be easily generalized to the general $r$-partition case.  
We briefly describe the generalizations below. 

\subsection{Structural entropy}
Our approach can deal with general SBMs similarly.
In a general SBM with $r\ge 2$ parts, an $r\times r$ symmetric matrix $P$ is used to describe the probabilities between and within the communities, where two vertices $u\in V_i$ and $v\in V_j$ are connected by an edge with probability $P_{ij}$ ($i$ and $j$ are not necessarily distinct).
To simplify the presentation, we only present the results below in its special form where $P_{ij}=p$ if $i=j$ and $P_{ij}=q$ if $i\neq j$, and we remark that similar results hold in the general case as well.
We first give the result on the computation of the partitioned structural entropy of SBM.
\begin{theorem}\label{thm:SBM-r}
	Fix $r$ reals $x_1, x_2, \dots, x_r$ in $(0,1)$ whose sum is 1.
	Let $V=V_1\cup V_2\cup \cdots \cup V_r$ be a set of $n$ vertices with a partition into $r$ parts such that $|V_i|=x_i n$.
	For large $n$ and all $p$ satisfying $n^{-1}\ln n\ll p, q$ and $1-p\gg n^{-1}\ln n$, the following holds:
	\begin{enumerate}
		\item[$(i)$] The $r$-partitioned structural entropy $H_{\mathcal S}^r$ for a partitioned structure $\mathcal S$ on $V$ is
		\begin{equation}
		H_{\mathcal S}^r = z (h(p)-h(q)) + \binom{n}{2} h(q) - r\log \left(\frac{n}{r}\right)! + O\left( \frac{\log n}{n^{\alpha}} \right)  \label{stru-entropy-r-value}
		\end{equation}
		for some $\alpha>0$, where $z:=\sum_{i=1}^r \binom {x_i n}{2}$.
		\item[$(ii)$] For a partitioned structure $S$ on $V$ and $\varepsilon >0$
		\[
		P\left( \left| -\frac{1}{\binom{n}{2}}\log P(S) - \frac{z}{\binom {n}{2}} (h(p)-h(q)) - h(q) + \frac{r\log (n/r)!}{\binom{n}{2}} \right| < 3\varepsilon \right) > 1 - 4\varepsilon.
		\]
	\end{enumerate}
\end{theorem}

\subsection{Compression algorithm}
The compression algorithm for a general $r$ with vertex partition $\{V_1,V_2,\dots,V_r\}$ can be viewed as a union of the compression algorithms for $S[V_i]$ and $S[V_i, V_j]$ ($i< j\in\{1,2,\dots,r\}$).
To be more precise, we describe the algorithm as follows. 
It first compresses all $S[V_i]$ into $\cL_i$ using \szip. 
Then run the \szip~decoder with input $\cL_i$ to obtain the decoded structure $S'[V_i]$. 
With the indices of $S'[V_i]$, $i=1,2, \dots, r$, we can compress $S[V_1,V_2,\dots, V_r]$ as a labeled $r$-partite graph into $\cL$ using an arithmetic encoder. 
This completes the encoding procedure and gives the codewords $\cL_1, \dots, \cL_r, \cL$, for which we concatenate together and get the final codeword. 
The decoding is to simply run the \szip~decoders and labeled (arithmetic) decoders parallelly. 
The correctness of the decoding output can also be argued accordingly. 

The performance of the algorithm can be obtained similar to \Cref{thm-alg-performance} as follows. 
\begin{theorem}\label{thm-alg-performance-r}
	Fix $r$ reals $x_1, x_2, \dots, x_r$ in $(0,1)$ whose sum is 1.
	Let $V=V_1\cup V_2\cup \cdots \cup V_r$ be a set of $n$ vertices with a partition into $r$ parts such that $|V_i|=x_i n$.
	Given a partitioned unlabeled graph $S$ on $V$, let $L(S)$ be the codeword length given by our algorithm.
	For large $n$, our algorithm runs in time $O(n^2)$, and satisfies the following:
	\begin{enumerate}[label=$(\roman*)$]
		\item The algorithm asymptotically achieves the structural entropy in~\eqref{stru-entropy-r-value}, i.e., 
		\begin{equation*}
		\bE[L(S)]\le \sum_{i=1}^r \binom {x_i n}{2} (h(p)-h(q)) + \binom{n}{2} h(q) - n\log n + O(n). \label{alg-length-r}
		\end{equation*}
		
		\item For any $\epsilon>0$,
		\begin{equation*}
		P(|L(S)-\bE[L(S)]|\leq \epsilon n\log n)\geq 1-o(1).
		\end{equation*}
		
	\end{enumerate}
\end{theorem}

\section{Conclusion}

In this paper we defined the partitioned unlabeled graphs and partitioned structural entropy, which generalize the structural entropy for unlabeled graphs introduced by Choi and Szpankowski~\cite{Choi-Wojciech-12}.
We then computed the partitioned structural entropy for Stochastic Block Models and gave a compression algorithm that asymptotically achieves this structural entropy limit.
As mentioned earlier, we believe that in appropriate contexts the structural information of a graph or network can be interpreted as a kind of semantic information, in which case, the communication schemes may benefit from structural compressions which considerably reduce the cost.

\bibliographystyle{abbrv}
\bibliography{Ref_SBM}
	
\end{document}